\documentclass{llncs}

\usepackage[T1]{fontenc}

\usepackage{amsmath}
\usepackage{amssymb}
\usepackage{cite}
\usepackage{thmtools}
\usepackage{hyperref,cleveref}

\title{Recognizing Numbers}
\author{Pranshu Gaba\inst{1}\orcidID{0009-0000-8012-780X} \and Arnab Sur\inst{2}\orcidID{0009-0007-5811-2864}}

\authorrunning{P. Gaba and A. Sur}

\institute{Tata Institute of Fundamental Research, Mumbai, India
\email{pranshu.gaba@tifr.res.in}\\
  \and Chennai Mathematical Institute, Chennai, India
  \email{arnabs@cmi.ac.in}
}

\begin{document}
\maketitle

\begin{abstract}
  The use of monoids in the study of word languages recognized by finite-state automata has been quite fruitful. In this work, we look at the same idea of “recognizability by finite monoids” for other monoids. In particular, we attempt to characterize recognizable subsets of various additive and multiplicative monoids over integers, rationals, reals, and complex numbers. While these recognizable sets satisfy properties such as closure under Boolean operations and inverse morphisms, they do not enjoy many of the nice properties that recognizable word languages do.

  \keywords{Recognizability \and Additive monoids \and Multiplicative monoids \and Algebraic automata theory}
\end{abstract}

\section{Introduction} \label{sec:intro}
Regular (or recognizable) subsets of the free monoid (i.e., the set of finite words over a finite alphabet) are one of the most well-studied classes of languages. Many seemingly unrelated notions --- recognizability by finite-state automata, recognizability by finite monoids or congruences of finite index, rationality, monadic-second order definability --- all coincide for free monoids~\cite{eilenberg1974automata,sakarovitch2009elements,pin2010mathematical}. The algebraic notion of recognizability by morphisms into finite monoids allows us to use tools from the structure theory of finite monoids which has been particularly useful to characterize certain subclasses of the recognizable languages. For instance, we have Schützenberger’s celebrated result that the star-free languages (languages described by rational expressions using only the union, concatenation, and complement operations) are exactly those recognized by aperiodic monoids (monoids which do not contain a group). Further, there are algebraic characterizations of piecewise and locally testable languages as well~\cite{eilenberg1974automata, pin2010mathematical}.

\paragraph*{Related work.}
There has been a lot of research to define algebraic recognizability for structures other than sets of finite words like sets of trees~\cite{gecseg2015tree}, infinite words~\cite{perrin1992infinite}, timed words~\cite{maler2004recognizable, bouyer2001algebraic}, and data words~\cite{bouyer2001algebraic,bojanczyk2013nominal}. The goal has been to show their equivalence with recognizability by automata models or definability by logics for these structures. However these ``recognizable sets'' are not recognized by finite monoids but by other algebras. The survey \cite{10.1007/978-3-540-28629-5_8} provides a great overview of algebraic recognizability of all the structures mentioned above.

Indeed, recognizability by finite monoids seems to be best suited only for finite words. However, we want to characterize the recognizable (by finite monoids) subsets of as many monoids as we can to figure out why they may or may not be interesting. For instance, the recognizable subsets of non-negative reals under addition ($\mathbb{R}_{\ge 0}$) are just $\emptyset, \mathbb{R}_{\ge 0}, \{0\}$, and $\mathbb{R}_{> 0}$ \cite{dima2001algebraic}. This fact indicates that one must look at other algebraic structures to define recognizability for timed languages.

\paragraph*{Contributions.}
In this work, we will primarily look at monoids of numbers: integers, rationals, reals, and complex numbers, each with addition and multiplication operations. 
For the additive monoids, we bring together known results and point out generalizations to arbitrarily divisible monoids.
To the best of our knowledge, the results for multiplicative monoids are new. 

For recognizability by finite monoids, monoids of numbers may seem very specific as only some of their more general facts like infinite generators or commutativity lend the characterizations we find, but these monoids illustrate them very well. The additive monoids of the naturals and the integers are the exceptions since they are special cases of the free monoid and the free group respectively.

\paragraph*{Outline.} In \Cref{sec:prelims}, we recall the definitions and notations that we shall be using for our work.
In \Cref{sec:additive-monoids}, we characterize the additive monoids 
\(\mathbb{Z}\),
\(\mathbb{Q}\),
\(\mathbb{R}\),
\(\mathbb{C}\),
\(\mathbb{Z}_{\ge 0}\),
\(\mathbb{Q}_{\ge 0}\),
and \(\mathbb{R}_{\ge 0}\),
Along the way, we show the triviality of recognizable subsets of a class of monoids that are \emph{arbitrarily divisible}.
In \Cref{sec:multiplicative-monoids}, we characterize the multiplicative monoids 
\(\mathbb{Z}\),
\(\mathbb{Q}\), 
\(\mathbb{R}\), and
\(\mathbb{C}\),
where we also show a useful property satisfied by recognizable subsets of (countable) infinitely generated free monoids.

\section{Preliminaries} \label{sec:prelims}
We recall some definitions and notations used in this text.

\paragraph*{Sets of numbers.}
We denote the set of integers by \(\mathbb{Z}\),  the set of rational numbers by \(\mathbb{Q}\), the set of real numbers by \(\mathbb{R}\), and the set of complex numbers by \(\mathbb{C}\).
In addition, for \(\mathbb{X} \in \{\mathbb{Z}, \mathbb{Q}, \mathbb{R}\}\), we denote non-negative numbers by \(\mathbb{X}_{\ge 0}\), positive numbers by \(\mathbb{X}_{> 0}\), and non-zero numbers by \(\mathbb{X} \setminus \{0\}\).
In particular, we also sometimes denote the natural numbers, that is the non-negative integers, by \(\mathbb{N}\). 

\paragraph*{Monoids.}
A \emph{monoid} $(M, \cdot_M, 1_M)$ is a set~$M$ with an associative binary operation~$\cdot_M$ and a unit element~$1_M$ for the binary operation, that is, $m \cdot_M (1_M) = m = (1_M) \cdot_M m$ for all $m \in M$.
We say that a monoid \((M, \cdot_M, 1_M)\) is finite (resp. infinite) if the underlying set \(M\) is finite (resp. infinite).
We omit the operation and the unit element when it is clear from context, that is, we simply use the symbol of the set $M$ to denote the monoid $(M, \cdot_M, 1_M)$.
We also omit the binary operation when clear from context, that is, we represent $m_1 \cdot_M m_2$ by $m_1 m_2$ and $m \cdot_M m $ by $m^2$.
We extend the monoid operation to sets $X, Y \subseteq M$ by defining $X \cdot_M Y = \{x \cdot_M y \mid x \in X, y \in Y\}$. The asterate (or Kleene-star) operation is then defined as 
\begin{equation*}
    X^* = \bigcup\limits_{i\ge 0} X^i
\end{equation*}
where $X^0 = \{1_M\} \text{ and } X^i = X \cdot_M X^{i-1} \text{ for } i > 0$.
Some examples of monoids are groups and free monoids.
A \emph{group} $(G, \cdot_G, 1_G)$ is a monoid in which every element has an inverse, that is, for all \(g \in G\), there exists a unique element \(g^{-1} \in G\) such that \(g \cdot_G g^{-1} = 1_{G} = g^{-1} \cdot_{G} g\).
The \emph{order} of a group $G$ is the number of elements in $G$ if $G$ is finite, or else, if $G$ is infinite, then the order of $G$ is also infinite.
Moreover, if \(G\) is a finite group, then for every element $g \in G$, the \emph{order of $g$} is the smallest positive integer $i$ such that $g^i = 1_G$.

\paragraph*{Idempotents.}
An element $e$ of a monoid $M$ is \emph{idempotent} if $e^2 = e$. 
In particular, the unit element $1_M$ of a monoid $M$ is idempotent.
The unit element $1_G$ of a group $G$ is the unique idempotent in the group $G$.
A monoid $M$ is \emph{idempotent} if all elements $m$ of $M$ are idempotent.

\paragraph*{Zero.} 
An element $z$ of a monoid $M$ is a \emph{zero} of $M$ if for all $m \in M$, we have that $zm = mz = z$. 
If a monoid contains a zero, then the zero is unique.
Given a monoid $M$ without a zero element, let $M_0 = M \cup \{0_M\}$ be the monoid $M$ appended with a zero element. If $M$ has a zero element, we denote the zero element by $0_M$ and it is the case that $M_0 = M$. In what follows, we assume that the zero is distinct from the unit, for otherwise the monoid is trivial.

\paragraph*{Roots.}
Given an element \(m\) of a monoid \(M\) and a positive integer \(k\), an element~\(n\) of \(M\) is called a \emph{\(k^{\text{th}}\) root} of \(m\) in \(M\) if \(n^k = m\). 
We denote by \(m_k\) a \(k^{\text{th}}\) root of \(m\).
In general, a \(k^{\text{th}}\) root of \(m\) may not exist, or if it exists, it need not be unique.

\paragraph*{Morphisms.}
A \emph{morphism} \(\varphi\) from a monoid $(M, \cdot_{M}, 1_{M})$ to a monoid $(N, \cdot_{N}, 1_{N})$ is a map $\varphi \colon M \to N$ that maps the unit element of \(M\) to the unit element of \(N\) and preserves the monoid operations.
Formally, we have that \(\varphi\) satisfies $\varphi(1_{M}) = 1_{N}$ and for all $m_1, m_2 \in M$, we have that $\varphi(m_1 \cdot_{M} m_2) = \varphi(m_1) \cdot_{N} \varphi(m_2)$.

\paragraph*{Recognizable sets.}
Given monoids $M$ and $N$, we say a subset $S$ of $M$ is \emph{recognized} by the monoid $N$ if there exists a morphism $\varphi \colon M \to N$ and a subset $T \subseteq N$, such that $\varphi^{-1}(T) = S$.
We say $S$ is \emph{recognizable} if it is recognized by a \emph{finite} monoid. 
We denote the set of all recognizable subsets of $M$ by $\text{Rec}(M)$.
We also assume that the recognizing morphisms are surjective.
Recognizable sets exhibit nice closure properties, which we summarize in \Cref{prop:closure-properties-rec}.
The results follow since inverse maps are well-behaved with respect to Boolean operations.
\begin{proposition} [\cite{eilenberg1974automata,pin2010mathematical,sakarovitch2009elements}]
    \label{prop:closure-properties-rec}
	Let $M$ be a monoid, and let $S_1, S_2 \in \text{Rec}(M)$ be recognizable subsets of $M$.
	Then, $S_1 \cup S_2, S_1 \cap S_2, M \setminus S_1 \in \text{Rec}(M)$.
	That is, recognizable sets of a monoid are closed under unions, intersections, and complements.
\end{proposition}

\section{Additive monoids} \label{sec:additive-monoids}

In this section, we look at the additive monoids of numbers. 
We list in \Cref{tab:additive-monoids-results} the recognizable subsets of the various monoids considered. 
The rest of the section is dedicated to showing these results, as well as some additional properties that apply to a wider class of monoids.

\begin{table}
  \centering
  \caption{Recognizable subsets of additive monoids \((\mathbb{X}, +, 0)\)}\label{tab:additive-monoids-results}
  \begin{tabular}{clcl}
    \hline
    \(\mathbb{X}\) & \(\text{Rec}(\mathbb{X})\) &\quad \(\mathbb{X}\) & \(\text{Rec}(\mathbb{X})\) \\\hline
    \(\mathbb{Z}\) & periodic sets &\quad \(\mathbb{Z}_{\ge 0}\) & ultimately periodic sets \\
    \(\mathbb{Q}\) & \(\{\emptyset, \mathbb{Q}\}\) &\quad \(\mathbb{Q}_{\ge 0}\) & \(\{\emptyset, \mathbb{Q}_{\geq 0}, \{0\}, \mathbb{Q}_{> 0}\}\)\\
    \(\mathbb{R}\) & \(\{\emptyset, \mathbb{R}\}\) &\quad \(\mathbb{R}_{\ge 0}\) & \(\{\emptyset, \mathbb{R}_{\geq 0}, \{0\}, \mathbb{R}_{> 0}\}\)\\
    \(\mathbb{C}\) & \(\{\emptyset, \mathbb{C}\}\)& & \\
    \hline
  \end{tabular}
\end{table}

\subsection{Non-negative integers $(\mathbb{Z}_{\ge 0}, +, 0)$}
A subset $X$ of $\mathbb{Z}_{\ge 0}$ is \emph{ultimately periodic} if there exists a positive integer $p$ (called a period) and  a positive integer $n_0$ such that for all integers $n > n_0$, we have that $n \in X$ if and only if $n + p \in X$. 
The recognizable subsets of \((\mathbb{Z}_{\ge 0}, +, 0)\) are precisely the ultimately periodic sets.

To see this, observe that \(\mathbb{Z}_{\ge 0}\) is isomorphic to the free monoid generated by a unary alphabet.
The recognizable subsets of free monoid with one generator are exactly the unary regular languages.
Thus, the recognizable subsets of $(\mathbb{Z}_{\ge 0}, +, 0)$ are the ultimately periodic sets. 
A detailed exposition is found in \cite{eilenberg1974automata}.

\subsection{Integers $(\mathbb{Z}, +, 0)$}
A subset $X$ of $\mathbb{Z}$ is \emph{periodic} if there exists a positive integer $p$ (called a period) such that for all $n \in \mathbb{Z}$, we have that $n \in X$ if and only if $n + p \in X$. 
The recognizable subsets of $(\mathbb{Z}, +, 0)$ are the periodic sets. 
We show this now.

The monoid $(\mathbb{Z}, +, 0)$ is a group, and since surjective morphisms map groups to groups, we have that any recognizable subset of \(\mathbb{Z}\) must be recognized by a finite group.
Let $G$ be a finite group and let $\varphi \colon \mathbb{Z} \to G$ be a morphism. 
This morphism is completely determined by $\varphi(1)$. 
Indeed, if $\varphi(1) = g$ for some $g \in G$, then 
we have $\varphi(0) = 1_G$, 
we have $\varphi(n) = g^n$ for all positive integers $n$, 
and we have $\varphi(n) = (g^{-1})^{|n|}$ for all negative integers $n$, where $g^{-1}$ is the inverse of $g$.
If $g$ has order $p$ in $G$, then for all subsets $H$ of $G$, we have that $\varphi^{-1}(H)$ is a periodic set with period $p$.

\subsection{Non-negative rationals $(\mathbb{Q}_{\geq 0}, +, 0)$}

To find the recognizable subsets of non-negative rationals, we show and use an observation that is applicable for a wider class of monoids.
We define a property of monoids, namely monoids that are \emph{arbitrarily divisible} and characterize their recognizable subsets. 
The special case of non-negative reals under addition was proven in \cite{dima2001algebraic}. 
Here, we generalize it for monoids with this property.

\begin{definition}
    A monoid $M$ is \emph{arbitrarily divisible} if for all $m \in M$ and positive integer $k$, there exists a \(k^{\text{th}}\) root of \(m\) in \(M\), that is, there exists $m_k \in M$ such that $(m_k)^k = m$.
  \label{def:arbitrarily-divisible}
\end{definition}

The following proposition shows that the image \(\varphi(M)\) of an arbitrarily divisible monoid \(M\) under a morphism \(\varphi\) 
is an idempotent monoid.
\begin{proposition}
	If $M$ is an arbitrarily divisible monoid, $N$ is a finite monoid, and $\varphi: M\to N$ is a morphism, then \(\varphi(M)\) is an idempotent monoid.
\label{prop:arbitrarily-divisible-monoid}
\end{proposition}

\begin{proof}
  Let $m \in M$. 
  We want to show that for all \(m \in M\) that \(\varphi(m)\) is an idempotent element, that is, \(\varphi(m)^2 = \varphi(m)\).
  If \(\varphi(m) = 1_N\), then we are done since the unit element \(1_N\) is an idempotent.
  Thus, the only interesting case is when $\varphi(m) \neq 1_N$.
  Let \(n \in N\) be the image of a $(|N|!)^{\text{th}}$ root of \(m\), that is, $\varphi\left(m_{|N|!}\right) = n$. Consider the sequence:
  \begin{equation*}
    n, n^2, n^3, \ldots n^{|N| + 1}
  \end{equation*}
  By the pigeonhole principle, there exist at least two elements in this sequence that are equal, say $n^i = n^j$ where $1 \le i < j \le |N|+1$. There is an idempotent $n^k$ where $i \le k < j$ and $(j-i) \mid k$ \cite{eilenberg1974automata}. Let $k' = \frac{|N|!}{k}$. Then we have $\varphi(m) = n^{|N|!} = n^{k \cdot k'} = (n^k)^{k'} = n^k$, which is an idempotent.
  \qed
\end{proof}

\begin{remark}
  Let $M$ be an arbitrarily divisible monoid. 
  From \Cref{prop:arbitrarily-divisible-monoid}, 
  it follows that for all \(m \in M\), it is the case that \(\varphi(m)\) is idempotent, and thus for all 
  positive integers $p$, we have that $ \varphi(m^p) = \varphi(m)$. 
  Moreover, for every $m \in M$, for all positive integers $q$, let $m_q$ denote a \(q^{\text{th}}\) root of \(m\), that is, $(m_q)^q = m$. 
  Then, again, from \Cref{prop:arbitrarily-divisible-monoid}, we have that $\varphi(m_q)$ is an idempotent, and thus $\varphi(m_q) = \varphi(m_q)^q = \varphi((m_q)^q) = \varphi(m)$. 
  Thus, for all positive integers $p, q$, we have $ \varphi((m_q)^p) = \varphi(m)^p = \varphi(m)$.
\label{rem:arbitrarily-divisible-monoid-rationals}
\end{remark}

Note that $(\mathbb{Q}_{\geq 0}, +, 0)$ is an arbitrarily divisible monoid since for every non-negative rational number $r$ and every positive integer $k$, there exists a non-negative rational number that is a $k^{\text{th}}$ root of $r$, namely $\frac{r}{k}$.
From \Cref{rem:arbitrarily-divisible-monoid-rationals}, it follows that for all positive rational numbers $r$, and for all morphisms to finite monoids, $\varphi: \mathbb{Q}_{\geq 0} \to M$, it is the case that $\varphi(r) = \varphi(1)$.
Thus, we have that the image of a surjective $\varphi$ has at most two elements $\varphi(0)$ and $\varphi(1)$, and we get that $\text{Rec}(\mathbb{Q}_{\geq 0}, +) = \{\emptyset, \{0\}, \mathbb{Q}_{> 0}, \mathbb{Q}_{\geq 0}\}$, since these sets are 
$\varphi^{-1}(\emptyset)$,
$\varphi^{-1}(\{\varphi(0)\})$,
$\varphi^{-1}(\{\varphi(1)\})$, and
$\varphi^{-1}(\{\varphi(0), \varphi(1)\})$ respectively.

\subsection{Rationals $(\mathbb{Q}, +, 0)$ and beyond}

Now, we look at the monoids $(\mathbb{Q}, +, 0)$, $(\mathbb{R}, +, 0)$, and $(\mathbb{C}, +, 0)$.
These are all arbitrarily divisible groups. 
The following proposition shows that the only recognizable subsets of arbitrarily divisible groups are the trivial sets.

\begin{proposition}
    If $M$ is an arbitrarily divisible group then $	\text{Rec}(M) = \{ \emptyset, M \}$.
\label{prop:Rec-arbitrarily-divisible-group-trivial}
\end{proposition}

\begin{proof}
We shall show that all morphisms $\varphi:M \rightarrow N$, where $N$ is a finite monoid, maps every element $m$ in \(M\) to the unit $1_{N}$ in $N$. 
To see this, observe that since $M$ is a group, $\varphi(M)$ is also a group. 
Let $k$ be the order of the group $\varphi(M)$. 
Since $M$ is arbitrarily divisible, there exists $m_{k}$ in $M$, a \(k^{\text{th}}\) root of \(m\). 
Thus, we have $\varphi(m) = \varphi((m_{k})^{k}) = {\varphi(m_{k})}^{k} = 1_{N}$.
  \qed
\end{proof}

\subsection{Non-negative reals $(\mathbb{R}_{\geq 0}, +, 0)$}

Let $\varphi: \mathbb{R}_{\geq 0} \to M$ be a surjective morphism to a finite monoid. Since $(\mathbb{R}_{\geq 0}, +, 0)$ is a commutative arbitrarily divisible monoid, by \ref{prop:arbitrarily-divisible-monoid} $M$ is a commutative idempotent monoid. Moreover by \Cref{rem:arbitrarily-divisible-monoid-rationals}, for each $q \in \mathbb{Q}_{> 0}$ we have $\varphi(q) = \varphi(1)$.

 Let $r$ be a positive irrational number. There exists $p, q \in \mathbb{Q}_{> 0}$ such that $p < r < q$. Then, $ \varphi(q) = \varphi(r) \varphi(q-r) = \varphi(2r) \varphi(q- r) = \varphi(r) \varphi(q) $. We also have that $ \varphi(r) = \varphi(p)\varphi(r-p) = \varphi(2p) \varphi(r-p) = \varphi(r)\varphi(p) = \varphi(r) \varphi(q) = \varphi(q)$. Thus every irrational maps to $\varphi(1)$ as well.

Therefore, there are only two possible choices of $\varphi$, either for all $r \in \mathbb{R}_{\geq 0},\ \varphi(r) = 1_M$, or for all $r \in \mathbb{R}_{>0}, \varphi(r) = \varphi(1)$ and $\varphi(0) = 1_M$. It follows that $\text{Rec}(\mathbb{R}_{\geq 0}, +, 0) = \{\emptyset, \{0\},  \mathbb{R}_{> 0}, \mathbb{R}_{\geq 0}\}$.

\section{Multiplicative monoids} \label{sec:multiplicative-monoids}

We are interested in finding the recognizable subsets of the multiplicative monoids \(\mathbb{Z}\), \(\mathbb{Q}\), \(\mathbb{R}\), and \(\mathbb{C}\).
We will first recall some observations that will simplify our study.
The first observation relates recognizable subsets of monoids with a zero element to recognizable subsets of monoids without a zero element.

\begin{proposition}
  Let $M$ be a monoid without a zero. Then $\text{Rec}(M_0)= \text{Rec}(M) \cup \{R \cup \{0_M\} \mid R \in \text{Rec}(M)\}$.
  \label{prop:monoids-with-zero-rec}
\end{proposition}
\begin{proof} 
  If $\varphi: M_0 \to N$ is a morphism to a finite monoid $N$, then $\varphi(0_M) = 0_N$. Indeed, for any $n \in N$ there exists $m \in \varphi^{-1}(n)$ (since we assume all recognizing morphisms are surjective), and $\varphi(0_M) \cdot n = \varphi(0_M) = n \cdot \varphi(0_M)$. Let $X = \varphi^{-1}(P)$ for some $P \subseteq N$ and $\psi:M \to N$ be the morphism defined by $\psi(m) = \varphi(m)$ for all $m \in M$. Then if $0_M \in X$ then $X \setminus \{0\}$ is recognized by $\psi$ and if $0_M \notin X$ then $X$ is recognized by $\psi$.
  \qed
\end{proof}
This helps us reduce our problem to finding the recognizable subsets of the multiplicative monoids 
\(\mathbb{Z} \setminus \{0\}\),
\(\mathbb{Q} \setminus \{0\}\),
\(\mathbb{R} \setminus \{0\}\), and
\(\mathbb{C} \setminus \{0\}\).
The next observation is concerned with the recognizable subsets of a direct product of monoids.

\begin{theorem}
  [Mezei \cite{pin2010mathematical}] Let $M_1,\ldots, M_n$ be monoids and let $M = M_1 \times \cdots \times
		M_n$. A subset of M is recognizable if and only if it is a finite union of subsets
	of the form $R_1 \times \cdots \times R_n$, where each $R_i$ is a recognizable subset of $M_i$.
 \label{thm:mezei}
\end{theorem}

\begin{remark}
    The monoids $(\mathbb{Z} \setminus \{0\}, \times, 1), (\mathbb{Q} \setminus \{0\}, \times, 1),(\mathbb{R}\setminus \{0\}, \times , 1)$ are all isomorphic to the direct products of $(\mathbb{Z}_{> 0}, \times, 1), (\mathbb{Q}_{> 0}, \times, 1),(\mathbb{R}_{> 0}, \times, 1)$ respectively with $\mathbb{Z}_2$. Note that $\text{Rec}(\mathbb{Z}_2) = 2^{\{0,1\}}$. \Cref{thm:mezei} thus gives us the recognizable subsets of these monoids from the recognizable subsets of the corresponding ``positive'' monoids.
  \label{rem:Rec-multiplicative-monoids-with-negatives}
\end{remark}

\begin{remark}
  Note that $(\mathbb{R}_{> 0}, \times, 1)$ and $(\mathbb{C} \setminus \{0\}, \times, 1)$ are both arbitrarily divisible groups and hence by \Cref{prop:Rec-arbitrarily-divisible-group-trivial} their recognizable subsets are trivial. By \Cref{rem:Rec-multiplicative-monoids-with-negatives} we have
  \begin{align*}
    \text{Rec}((\mathbb{R}, \times, 1)) &= \{\emptyset, \{0\}, \mathbb{R}_{<0}, \mathbb{R}_{\le 0}, \mathbb{R}_{>0}, \mathbb{R}_{\ge 0}, \mathbb{R} \setminus \{0\}, \mathbb{R}\}\\
    \text{Rec}((\mathbb{C}, \times, 1)) &= \{\emptyset, \{0\},  \mathbb{C}  \setminus \{0\}, \mathbb{C}\}.
  \end{align*}
  \label{rem:rec-reals-and-complex-under-multipication}
\end{remark}

In the rest of this section, we work towards finding the recognizable subsets of the multiplicative monoids $\mathbb{Z}_{>0}$ and $\mathbb{Q}_{>0}$.
We shall look at recognizable subsets of (countable) infinitely generated free monoids. Indeed naturals and rationals with multiplication are merely the commutative infinitely generated free monoid and the commutative infinitely generated free group respectively. We show that a recognizing morphism can be ``factorized'' into two morphisms via a finitely generated free monoid that also recognizes the set.

\begin{lemma}
  Let $A$ be an infinite alphabet and $\varphi: A^* \to M$ be a morphism to a finite monoid $M$. Then, there exists an equivalence relation $\sim$ on $A$ of finite index $k$ and a morphism $\psi: (A/\sim)^* \to M$ such that $\varphi = \psi \circ \pi_\sim$ where $\pi_\sim: A^* \to (A/\sim)^*$ is the projection morphism of $\sim$. 
  \label{lem:infinitely-generated-rec-morphism}
\end{lemma}
\begin{proof}
  We define the equivalence relation $\sim$, for all $a, a' \in A$, $a \sim a'$ iff $\varphi(a) = \varphi(a')$. Note that $\sim$ has a finite index. We define the morphism $\psi: (A/\sim)^* \to M$ by $\psi([a]_\sim) = \varphi(a)$ for $a \in A$. It is clearly well-defined and is a morphism.
  We now show that $\varphi = \psi \circ \pi_\sim$.
  Let $w = a_1a_2\dots a_p \in A^*$. Then,
  \begin{align*}
    \psi(\pi_\sim(a_1 \dots a_p)) &= \psi([a_1]_\sim [a_2]_\sim \dots [a_p]_\sim) \\
    &= \varphi(a_1)\varphi(a_2) \dots \varphi(a_p) = \varphi(w).
  \end{align*}
  \qed
\end{proof}

\begin{proposition}
    If $X \in \text{Rec}(A^*)$ then, there exists an equivalence relation $\sim$ on $A$ of finite index $k$ such that $X = (\pi_\sim^{-1} \circ \pi_\sim)(X)$ where $\pi_\sim: A^* \to (A/\sim)^*$ is the projection morphism of $\sim$.
  \label{prop:pi-tilde-rec}
\end{proposition}

\begin{proof}
  Let $X$ be recognized by the surjective morphism $\varphi: A^* \to M$ where $M$ is a finite monoid. We have by \Cref{lem:infinitely-generated-rec-morphism} that there exists an equivalence relation $\sim$ on $A$ of finite index $k$ and a morphism $\psi: (A/\sim)^* \to M$ such that $\varphi = \psi \circ \pi_\sim$ where $\pi_\sim: A^* \to (A/\sim)^*$ is the projection morphism of $\sim$.
  Clearly, $X \subseteq (\pi_\sim^{-1} \circ \pi_\sim)(X)$. Assume on the contrary there is some $w \in (\pi_\sim^{-1} \circ \pi_\sim)(X) \setminus X$. 
  Then there exists $w' \in X$ such that $\pi_\sim(w) = \pi_\sim(w')$. 
  Then, $\varphi(w') = \psi(\pi_\sim(w')) = \psi(\pi_\sim(w)) = \varphi(w)$, and thus $w \in X$ which is a contradiction.
  \qed
\end{proof}
We would like to mention a way to define $\text{Rec}(A^*)$ in monadic-second order logic using \Cref{prop:pi-tilde-rec}. Given $X \in \text{Rec}(A^*)$, let $\sim$ be the equivalence relation of finite index $k$ given by \Cref{prop:pi-tilde-rec}. Now by \Cref{lem:infinitely-generated-rec-morphism}, $\pi_\sim(X) \in \text{Rec}((A/\sim)^*)$, and therefore $\pi_\sim(X)$ is a model for a monadic-second order (MSO) formula with first-order variables interpreted as positions in the word and second-order variables interpreted as sets of position, and unary predicates $R_b$ for each $b \in A/\sim$ ($R_b(i)$ if and only if the $i^\text{th}$ letter is $b$) and binary predicate $<$ ($i<j$ if position $i$ appears before position $j$)~\cite{pin2010mathematical}. The set $X$ is a model for this MSO formula if the predicates $R_b$ are interpreted as $R_b(i)$ if and only if the $i^{th}$ letter is $a$ such that $\pi_\sim(a) = b$. This is unsatisfactory since this is not a uniform characterization and much of the computational difficulty is hidden in the interpretation of the predicates $R_b$. Indeed, the predicates $R_b$ may not even be computable as the following remark shows.

\begin{remark}
    If $A$ is a countably infinite alphabet then $\text{Rec}(A^*)$ is uncountable. Consider the monoid $M_3 = \langle \{0_M, 1_M, p\} \mid p^2 = 0_M \rangle$ where $0_M$ is the zero, $1_M$ is the unit. Any $X \subseteq A$ is recognized by the morphism $\varphi: A^* \to M_3$ defined by $\varphi(a) = p$ for $a \in X$ and $\varphi(a) = 0_M$ otherwise. Thus, one cannot hope for a characterization by a finitely describable machine or logic.
    \label{rem:rec-infinite-generated-free-monoid-uncountable}
\end{remark}

Despite the large number of recognizable subsets, we show that many seemingly easy-to-compute subsets of naturals and rationals under multiplication are not recognizable.

\subsection{Positive integers $(\mathbb{N} \setminus \{0\}, \times, 1)$}

Working directly with $(\mathbb{N} \setminus \{0\}, \times, 1)$ is a bit difficult, so we choose instead to work with an isomorphic monoid, the monoid of finite sequences of naturals that comes from the powers of the primes in a natural number's prime factorization.

Let $\sigma$ and $\tau$ be finite sequences of non-negative integers of length $m$ and $n$ respectively. Without loss of generality we assume that $m \le n$. The pointwise addition of $\sigma$ and $\tau$, $\sigma + \tau$ is the sequence of length $n$ defined by
\begin{equation*}
    (\sigma + \tau)(i) = 
    \begin{cases}
        \sigma(i) + \tau(i) &\text{if } i \le m, \\
        \tau(i) &\text{otherwise}
    \end{cases}
\end{equation*}
Finite sequences of non-negative integers form a monoid under pointwise addition with the sequence of length 0 as the unit element. We define for each $n \in \mathbb{N}$, the sequence of non-negative integers of length $n$, $\sigma_n$ where for each $1 \le i < n, \sigma_n(i) = 0$ and $\sigma_n(n) = 1$. Let $\sigma_0$ be the sequence of length 0 and $N = \{\sigma_n\}_{n \in \mathbb{N}}$. $N$ generates the monoid of finite sequences of non-negative integers under pointwise addition and thus we denote the monoid by $N^*$.
 
It is easy to see that $(\mathbb{N} \setminus \{0\}, \times, 1)$ is isomorphic to $N^*$. Indeed, by the fundamental theorem of arithmetic, one can associate each natural with the sequence of powers in its unique prime factorization. Thus, we shall focus on $\text{Rec}(N^*)$. 

Let $\varphi: A^* \to M$ be a morphism to a finite monoid $M$. . Then, by \Cref{lem:infinitely-generated-rec-morphism} there exists an equivalence relation $\sim$ on $N$ of finite index $k$ and a morphism $\psi: (N/\sim)^* \to M$ such that $\varphi = \psi \circ \pi_\sim$ where $\pi_\sim: N^* \to (N/\sim)^*$ is the projection morphism of $\sim$. 

By the commutativity of $N^*$, the monoid $(N/\sim)^*$ is a commutative free monoid generated by $k$ elements and is hence isomorphic to $\mathbb{N}^k$ ($k$-tuples of non-negative integers) under pointwise addition \cite{sakarovitch2009elements}. Fix some enumeration $N/\sim = \{\alpha_1, \dots \alpha_k\}$. We then have
\begin{equation*}
  \pi_\sim(\sigma) = \Big(\sum\limits_{i \in \alpha_j} \sigma(i)\Big)_{j \in [k]}
\end{equation*}
Hereafter we shall denote elements of $\pi_\sim(N^*)$ by elements of $\mathbb{N}^k$ as defined above.

\subsubsection{Applications of \Cref{prop:pi-tilde-rec}}

Using \Cref{prop:pi-tilde-rec} we can show certain subsets of $N^*$ are not recognizable. We call a property $\mathcal{P}$ of $\mathbb{N}$, \emph{trivial} if the set $P = \{n \in \mathbb{N} \mid \mathcal{P}(n)\}$ is finite or co-finite. For $\sigma \in N^*$ we denote by $|\sigma|$ the smallest natural number $n_0$ such that for all $n > n_0$, $\sigma(n) = 0$. Examples of non-trivial properties are parity and primaility.

\begin{proposition} The following sets are not recognizable for every non-trivial property $\mathcal{P}$ of $\mathbb{N}$.
  \begin{enumerate}
    \item $X = \{\sigma \in N^* \mid \mathcal{P}(|\sigma|)\}$, 
    \item $X = \{\sigma \in N^* \mid \forall i \in \mathbb{N}, \mathcal{P}(\sigma(i))\}$, and 
  \item $X = \{\sigma \in N^* \mid \exists i \in \mathbb{N}, \mathcal{P}(\sigma(i))\}$.   \end{enumerate}
  \label{prop:non-recognizable-subsets-of-N*}
\end{proposition}

\begin{proof}
  The proof template is the same for all three parts. We assume on the contrary there exists a finite monoid $M$ and morphism $\varphi: N^* \to M$ that recognizes $X$. By \Cref{prop:pi-tilde-rec} there exists an equivalence relation $\sim$ on $N$ of finite index $k$ such that $X = (\pi_\sim^{-1} \circ \pi_\sim) (X)$. Since $\sim$ is of finite index and $\mathcal{P}$ is not trivial, there exist natural numbers $n_1 < n_2 < n_3$ such that $\mathcal{P}(n_1)$ and $\mathcal{P}(n_3)$ but $\neg\mathcal{P}(n_2)$ holds and $\sigma_{n_1} \sim \sigma_{n_3}$. We show two sequences, $s_1 \in X$ and $s_2 \notin X$, but $\pi_\sim(s_1) = \pi_\sim(s_2)$ for each of the three parts.
  \begin{enumerate}
    \item  $s_1 = \sigma_{n_2} + \sigma_{n_3} \in X$ and $s_2 = \sigma_{n_1} + \sigma_{n_2} \notin X$ but $\pi_\sim(s_1) = \pi_\sim(s_2)$.

    \item Let $s_1 \in X$ and $\sigma_i \sim \sigma_j$ such that $s_1(i) = n_1$ and $s_1(j) = n_3$. We construct the sequence $s_2$ such that for all $l \neq i, j$, $s_2(l) = s_1(l)$, $s_2(i) = n_2$ and $s_2(j) = n_1 + n_3 - n_2$. Clearly $s_2 \notin X$ but $\pi_\sim(s_1) = \pi_\sim(s_2)$.

    \item If $X \in \text{Rec}(N^*)$ then $X^c \in \text{Rec}(N^*)$. However, $X^c = \{\sigma \in N^* \mid \forall i \in \mathbb{N}, \neg\mathcal{P}(\sigma(i))\}$ and $\neg\mathcal{P}$ is also non-trivial, contradicting the observation above.
  \end{enumerate}
  \qed
\end{proof}

Thus sequences of even naturals (which translates to perfect squares for the monoid of naturals under multiplication), or sequences of even length (naturals divisible by exactly an even number of primes) are not recognizable despite seeming to be easy computable by some finite state automaton that, say, takes input sequences encoded in unary.

Since $N^*$ is the free commutative monoid generated by a countably infinite set and the monoid $M_3$ defined in \Cref{rem:rec-infinite-generated-free-monoid-uncountable} is also commutative, we conclude that $\text{Rec}(N^*)$ is uncountable. We cannot expect a nice machine or logic characterization for $\text{Rec}(N^*)$.

\subsection{Positive rationals $(\mathbb{Q}_{>0}, \times, 1)$}

Recall that we looked at finite sequences of naturals with pointwise addition instead of positive naturals with multiplication. For this case, we shall look at the monoid of finite sequences of integers with pointwise addition instead.

Finite sequences of integers form a monoid under pointwise addition with the sequence of length 0 as the unit element. We define for each $n \in \mathbb{Z}$, the sequence of integers of length $|n|$, $\sigma_n$ where for each $1 \le i < |n|$, we have $\sigma_n(i) = 0$ and $\sigma_{n}(|n|) = n/|n|$. Let $\sigma_0$ be the sequence of length 0 and $Z = \{\sigma_n\}_{n \in \mathbb{Z}}$. $Z$ generates the monoid of finite sequences of integers under pointwise addition and thus we denote the monoid by $Z^*$.

Let $\varphi: A^* \to M$ be a morphism to a finite monoid $M$. .Then, by \Cref{lem:infinitely-generated-rec-morphism} there exists an equivalence relation $\sim$ on $Z$ of finite index $k$ and a morphism $\psi: (Z/\sim)^* \to M$ such that $\varphi = \psi \circ \pi_\sim$ where $\pi_\sim: Z^* \to (Z/\sim)^*$ is the projection morphism of $\sim$. 

Note that by the commutativity of $Z^*$, the monoid $(Z/\sim)^*$ is a commutative free monoid generated by $k$ elements and is hence isomorphic to $\mathbb{Z}^k$ ($k$-tuples of integers) under pointwise addition. Fix some enumeration $Z/\sim = \{\alpha_1, \dots \alpha_k\}$. We then have
\begin{equation*}
  \pi_\sim(\sigma) = \Big(\sum\limits_{i \in \alpha_j} \sigma(i)\Big)_{j \in [k]}
\end{equation*}
Hereafter we shall denote elements of $\pi_\sim(Z^*)$ by elements of $\mathbb{Z}^k$ as defined above.

\subsubsection{Applications of the \Cref{prop:pi-tilde-rec}}

The next proposition shows that every non-empty recognizable subset $X$ of $Z^*$ contains sequences of length greater than $p$ for all $p > 0$.
A consequence of the following proposition is that in every recognizable subset of $(\mathbb{Q}_{>0}, \times, 1)$, the set of prime factors of the elements of the subset is infinite.
\begin{proposition}
  If $X \in \text{Rec}(Z^*)$ and $X \neq \emptyset$, then for all integers $p > 0$, there exists a sequence $\sigma \in X$ such that $|\sigma| > p$.
  \label{prop:Rec-Z*-sequence-lenth-unbounded}
\end{proposition}

\begin{proof}
  Let $\sigma \in X$ and $p \in \mathbb{N}$. If $|\sigma| > p$ then we are done. Assume $|\sigma| \le p$. Then there exists $n, n' > p$ such that $\sigma_n \sim \sigma_{n'}$. Note that $\sigma(n) = \sigma(n') = 0$. We construct the sequence $\sigma'$ where $\sigma'(n) = 1$, $\sigma'(n') = -1$, and for all $m \neq n, n'$, $\sigma'(m) = \sigma(m)$. Clearly, $|\sigma'| > p$ and $\pi_\sim(\sigma) = \pi_\sim(\sigma')$, and thus $\sigma' \in X$.
  \qed
\end{proof}

The proof of the following proposition follows the same idea in \Cref{prop:non-recognizable-subsets-of-N*} 

\begin{proposition}
  The following sets are not recognizable for every non-trivial property $\mathcal{P}$ of $\mathbb{Z}$.
  \begin{enumerate}
    \item $X = \{\sigma \in Z^* \mid \forall i \in \mathbb{N}, \mathcal{P}(\sigma(i))\}$, and 
  \item $X = \{\sigma \in Z^* \mid \exists i \in \mathbb{N}, \mathcal{P}(\sigma(i))\}$.
  \end{enumerate}
  \label{prop:Z*-unrecognizable}
\end{proposition}

\Cref{rem:rec-infinite-generated-free-monoid-uncountable} doesn't apply directly here as it did for $N^*$ since $M_3$ is not a group but $Z^*$ is. However indeed, $\text{Rec}(Z^*)$ is uncountable. Given any arbitrary subset $X \subseteq \mathbb{Z}$ we construct the set
\begin{equation*}
  S(X) = \{\sigma \in Z^* \mid \sum_{n \in X} \sigma(|n|) = 1\ \land \sum_{n \notin X} \sigma(|n|) = 0\}.
\end{equation*}
It is not difficult to see that $S(X) = S(Y)$ iff $X = Y$ for any $X, Y \subseteq Z$. Indeed, each $S(X) \in \text{Rec}(Z^*)$. Define $\sigma_n \sim \sigma_m$ iff $n, m \in X$. Thus $\pi_\sim(S(X)) = (0, 1)$ which is trivially recognizable in $Z^2$. Thus we have the same conclusion for $\text{Rec}(Z^*)$: we cannot expect to get a nice machine or logic characterization for $\text{Rec}(Z^*)$.

\section{Conclusion} \label{sec-conclusion}
We looked at the recognizable subsets of the additive and multiplicative monoids of integers, rationals, reals, and more. The arbitrarily divisible property of some of the monoids implied they have rather trivial recognizable sets. Hence recognizability by finite monoids cannot be the right notion of algebraic recognizability. On the other hand, for integers and rationals with multiplication we encounter a more complicated problem. They have an uncountable number of recognizable subsets and many computable subsets are not recognizable. 

It would be interesting to look at monoids of matrices under multiplication or in general monoids of relations or functions over the domains considered. 
Our general results above do not translate well or immediately show us the recognizable subsets of these monoids since they have different structures, especially since they may not be arbitrarily divisible and their generators may not be easily identified. 
Further, recognizable sets of relations, or transductions, also have been extensively studied for the case of relations over free monoids \cite{eilenberg1974automata, sakarovitch2009elements}. 
It might be interesting to see if the observations in the present work also show up in the case of relations over rationals and reals.

\bibliographystyle{splncs04}
\bibliography{rec-zoo}
\end{document}